\documentclass[final]{siamltex}

\usepackage{amsmath, amssymb,epsfig , algorithm, algorithmicx}
\usepackage{algpseudocode, url, pifont}

\newenvironment{changemargin}[2]{%
  \begin{list}{}{%
    \setlength{\topsep}{0pt}%
    \setlength{\leftmargin}{#1}%
    \setlength{\rightmargin}{#2}%
    \setlength{\listparindent}{\parindent}%
    \setlength{\itemindent}{\parindent}%
    \setlength{\parsep}{\parskip}%
  }%
  \item[]}{\end{list}}

\title{A Message-Passing Algorithm for Graph Isomorphism}
\author{Mohamed F. Mansour\thanks{The author is with Amazon  Inc.,  Sunnyvale, CA, 94089.}}

\begin{document}

\maketitle

\begin{abstract}
A message-passing procedure for solving the graph isomorphism problem  is proposed. The procedure resembles the belief-propagation algorithm in the context of graphical models inference and LDPC decoding. To enable the algorithm, the input  graphs are transformed into intermediate canonical representations of bipartite graphs. The matching procedure injects specially designed input patterns to the canonical graphs and runs a message-passing algorithm to generate two output fingerprints that are matched if and only if the input graphs are isomorphic.

\end{abstract}



\pagestyle{myheadings}
\thispagestyle{plain}
\markboth{Mansour:A Message-Passing Algorithm for Graph Isomorphism}
{Mansour:A Message-Passing Algorithm for Graph Isomorphism}

\section{Introduction}
Graph isomorphism  is a classical problem at the intersection of few disciplines including graph theory, pattern recognition, and computing theory; 
with both theoretical and applied importance. At the applied side, it is a special case of graph matching, which is a cornerstone of many pattern recognition applications. At the theoretical side, it is one of the few problems in computing theory whose complexity is not known \cite{cormen2009introduction, graph_book}. A wealth of research work has addressed this problem during the last few decades, and several tutorials and workshops have been devoted to the problem \cite{thirty_years, kobler2012graph, babai2016graph, foggia2014graph}. In section \ref{sec:graph_background}, we give a brief overview of relevant prior art in the subject. 

In this work,  a new algorithm for graph isomorphism is developed. 
The algorithm adapts the well-known belief propagation algorithm \cite{Koller} to generate signatures that reflect the edge structure, wherein two graphs are isomorphic if and only if their signatures are identical. To enable the belief propagation algorithm for a general graph, the input graph is transformed to a bipartite canonical representation that preserves the edge structure. 
The matching algorithm is an iterative procedure that progressively matches pairs of nodes from the two canonical graphs by comparing their signatures when excited by a specially designed input pattern. The salient feature of the algorithm
is the absence of backtracking during the matching procedure. The conditions for the completeness of the proposed algorithm are derived, and an efficient implementation that resembles the sum-product algorithm is presented. The effectiveness of the algorithm is established by evaluating it using the TC-15 graph database \cite{foggia2001database}.

The paper is organized as follows. In section \ref{sec:background} a brief overview of  prior art of the graph isomorphism problem and the belief propagation algorithm is introduced. The proposed algorithm is described in details in section \ref{sec:algorithm}. The completeness conditions of the algorithm are established in section \ref{sec:analysis}. Finally, in section \ref{sec:discussion}, we provide generalizations of the proposed algorithm to other relevant graph matching problems, and present the evaluation results.

\subsection{Notations}
The discussion in this work assumes undirected and unweighted graphs. A generalization to other graph types is described in section \ref{sec:generalization}. 
The following notations are used:
\begin {itemize}
\item G1 and G2 refer to the two graphs under test.
\item $M$ is the total number of nodes in the graph.
\item $K$ is the total number of edges in the graph.
\item $L$ is the longest simple route in the graph.
\item bold letters/symbols refer to sets (or vectors), while non-bold ones refer to individual entries. 
\item $\Phi(.) : \mathbb{R}^M  \rightarrow \mathbb{R}^{ML}$ is the transformation that converts an input pattern of length $M$ to an output signature of length $ML$.
\item $\bf{A}$ is the set of ordered pairs of matched nodes.
\item $\overline{{\bf{A}}}_1$ and  $\overline{{\bf{A}}}_2$ are the sets of nodes in $G1$ and $G2$ respectively that have not been matched.
\item A node $\nu_i \in G2$ is \emph{similar} to a node $\nu_k \in G1$ (denoted $\nu_i \sim \nu_k$) if they belong to analogous cells in the two graphs, where a graph cell is one set of the graph partition to be defined in section \ref{sec:canon}.
\item ${\bf{Q}}(v)$ denotes the set of edges connected to node $v$. The order of ${\bf{Q}}(v)$, $|{\bf{Q}}(v)|$, is the number of edges connected to $v$.
\item ${\bf{Q}}(v)-\{E\}$ denotes the complement of $\{E\}$ in ${\bf{Q}}(v)$, i.e., edges other than $E$ that are connected to $v$.
\item If $|{\bf{Q}}(v)|=2$, then $\tilde{E}$ denotes the complement edge of $E$ in ${\bf{Q}}(v)$, i.e., ${\bf{Q}}(v)-\{E\} = \{\tilde{E}\}$.
\end {itemize} 

\section{Background}\label{sec:background}

\subsection{The Graph Isomorphism Problem} \label{sec:graph_background}
An isomorphism between two graphs is a bijection between the nodes of the two graphs that preserves the edge structure \cite{graph_book}. Graph automorphism is an identical problem for finding an isomorphism between a graph and itself, rather than the trivial identity mapping. Subgraph isomorphism is a related problem of finding an isomorphism between a graph and subgraph of another graph. These  problems are under the class of exact graph matching, where the edge structure is preserved. Inexact graph matching is more relevant in many applications, where the objective is to find a mapping that minimizes a distortion metric between the two graphs.

The graph isomorphism problem has been  studied by both applied and theoretical researchers \cite{babai2016graph}. This work falls into the applied research category of graph isomorphism. As noted in \cite{thirty_years, foggia2014graph},  most of the algorithms in this category can be  classified into two broad classes: tree-based algorithms, and algorithms based on canonical representation. Tree-based algorithms, e.g.,  \cite{ullmann1976algorithm, ghahraman1980graph,cordella1998graph,cordella2004sub,larrosa2002constraint}, progressively matches pairs of nodes as a tree search, with  backtracking when the search reaches a dead-end. Several data structures and heuristics were proposed to optimize the implementation of the basic procedure in \cite{ullmann1976algorithm}. The second category of practical graph isomorphism algorithms, e.g. \cite{mckay1981practical, mckay2014practical, darga2004exploiting, darga2008faster,junttila2011conflict},  maps the graph to a canonical representation that preserves isomorphism. The matching of the graphs is performed on the canonical representation which has polynomial time complexity. However, the conversion to the canonical form can have exponential complexity in the worst case \cite{miyazaki1997complexity}. 

A more relevant class of algorithms to our work is the class of algorithms based on randow walks, e.g., \cite{RW2005, cho2010reweighted, douglas2008classical}. The general idea is to generate a graph signature from the steady state output of the random walk, and match the signatures of the graphs under test. For example, in \cite{RW2005}, the steady state probability distribution of the random walk is used as a signature, and a random walk that resembles the  page rank algorithm \cite{page1999pagerank} is shown to provide a polynomial-time algorithm for a large class of graphs but fails under certain conditions on the graph spectrum.

The proposed matching algorithm combines the idea of canonical representations with the generation of a signature that reflects the edge structure. Unlike the canonical approach in \cite{mckay1981practical}, the canonical representation of the proposed algorithm is only an intermediate step in the graph matching algorithm. The transformation to the canonical form has polynomial complexity, but non-isomorphic graphs can have identical canonical representations. The canonical representation defines, by construction,  a partitioning of a graph that could  reveal most non-isomorphic graphs, and hence simplifies the matching procedure. The matching procedure in the proposed algorithm utilizes the graph partitioning in the canonical representation to prune the search space. This partitioning is also exploited to design the input patten to the message-passing algorithm. The process of generating the output signature is inspired by the message-passing algorithm in LDPC decoding \cite{Lin}, which is described in more details in the following subsection.

\subsection{Belief-Propagation Algorithm} \label{sec:MP}
The belief-propagation algorithm was developed for inference over graphical models of graphs with no cycles, i.e., trees \cite{Koller}, but few generalizations were introduced to graphs with cycles. In the context of LDPC decoding, a version of loopy belief-propagation algorithm, the sum-product algorithm, is utilized to compute the decoding posteriors \cite{Richardson, Lin}. 
In a nutshell, the decoder input is a vector of the likelihoods of each encoded bit after the communication channel. The output of the decoder is the likelihood of the decoded bits after processing the input likelihoods by running a message-passing algorithm on a graph that is designed from the LDPC code parity check matrix. The output likelihoods are, in general,  uniquely determined for a given input distribution and code matrix. 
A bipartite representation of LDPC parity check matrix, a.k.a, Tanner graph \cite{tanner}, is used. The Tanner graph is composed of variable nodes (that represent code bits) and check nodes (that represent parity bits).  Messages are exchanged between the two types of nodes, where $r_{ji}$ is a message from the $j$-th check-node to $i$-th variable node, and $q_{ij}$ is a message from the $i$-th variable node to the $j$-th check node. After initializing the messages with the likelihoods at the channel ouptut, the message passing algorithm updates the messages iteratively as follows \cite{Lin}:
\begin{eqnarray}
r_{ji}(0) &=& \frac{1}{2} + \frac{1}{2}  \prod_{i^\prime \in V_j-\{i\}} (1-2q_{i^\prime j}(1)) \\
r_{ji}(1) &=& 1-r_{ji}(0)\\
q_{ij}(0) &=& K_{ij} (1-P_i) \prod_{j^\prime \in C_i-\{j\}} r_{j^\prime i}(0) \\
q_{ij}(1) &=& K_{ij} P_i \prod_{j^\prime \in C_i-\{j\}} r_{j^\prime i}(1)
\end{eqnarray}
where $P_i$ is computed from the channel likelihood, $V_j$ is a set of variable-nodes connected to the check-node $j$, $C_i$ is the set of check-nodes connected to the variable-node $i$, and $K_{ij}$ is a normalization factor. In practice, the likelihood is computed in the log-domain, and the product in the above relations is converted to a sum operator, which is more tractable numerically. The essence of the algorithm that is used in the proposed matching algorithm is that a message from node $i$ to node $j$ combines messages from all other nodes $\{j^\prime\}$ that are connected to $i$. 
In practice, the algorithm converges to the output posteriors that depend on the input and the code graph after few decoding iterations.
 
\section{Matching Algorithm}\label{sec:algorithm}
\subsection{Canonical Representation}\label{sec:canon}

To enable the utilization of the belief-propagation algorithm, a \emph{bipartite} graph representation, that preserves graph isomorphism, is needed. The two disjoint sets of nodes of the canonical bipartite graph represent the nodes and the edges of the original graph. Each node in the first set corresponds to a node in the original graph, and each node in the second set corresponds to an edge in the original graph. Denote a node in the first set by $\nu$ and a node in the second set by $\xi$. An edge $\nu_i \text{--} \xi_k$ in the bipartite graph exists if and only if $\xi_k$ represents an edge in the original graph that is connected to the node in the original graph corresponds to $\nu_i$. Hence, each $\xi_k$ in the bipartite graph has exactly two edges to the nodes that  represent the bounding nodes in the original graph. Thus, the number of edges in the canonical graph is  twice the number of edges in the original graph.

A simple sorting procedure is applied to the nodes of the bipartite graph. $\{\nu_k\}$ are sorted according to their order. 
If two or more nodes have the same order, they are sorted chronologically according to the order of their neighbors. For example, If $v_1$ and $v_2$ have the same order, but $v_2$ has a neighbor with higher order than any neighbor of $v_1$, then $v_2$ is ranked first in the bipartite graph. If the highest neighbor orders of the vertices are the same, then the second highest orders of the neighbors are compared and so on. If the orders of all neighbors are the same, then the corresponding vertices are placed in \emph{arbitrary} order relative to each others in the bipartite graph. These are denoted as \emph{similar} nodes, and a class of similar nodes is denoted as a \emph{cell}. 
For consistency, a node that is not similar to any other node is also denoted as a cell of size $1$. 
Note that, the cells of a graph define a partition of the bipartite representation that is exploited by the matching algorithm where only pairs of nodes in analogous cells  are investigated for matching. 
An example of this canonical representation is shown in Fig. \ref{fig:canon_example}. 
\begin{figure}[ht]
\begin{center}
\includegraphics[width=4.0in,height=1.3in, trim = 0.5in 4.1in 6in 1.9in, clip,]{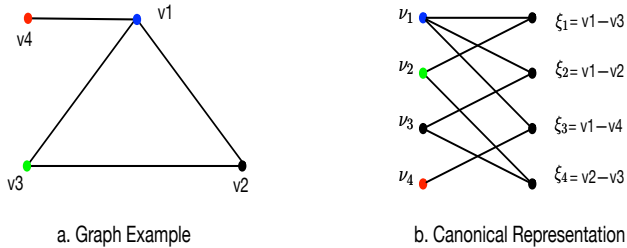}
\end{center} 
\caption{Example of the canonical representation. The bipartite graph has 3 cells: $\{\nu_1\}$, $\{\nu_2, \nu_3\}$, $\{\nu_4\}$}.
\label{fig:canon_example}
\end{figure}

The edge structure of the original graph is preserved in the canonical representation. Hence,  matching the bipartite canonical graphs is equivalent to matching the original graphs.
Note that, canonical representations of isomorphic graphs are not necessarily identical, because of the arbitrary ordering of nodes within a cell. Nevertheless, it significantly simplifies the matching algorithm of most graphs by restricting the matching only to analogous cells of the  canonical graphs. 
 
\subsection{Algorithm Outline} \label{sec:outline}
Before running the matching procedure,  the two graphs are converted to the canonical form as described in the previous subsection. A necessary condition for isomorphism is that the canonical representations have identical cell structures. Therefore, it is assumed in the following discussion that the graphs are connected and have the same cell structure.

After converting the two graphs to canonical forms, a message-passing algorithm is applied repeatedly to a specially designed input pattern to generate an output signature that reflects the graph edge structure. Isomorphic graph would produce identical signatures and vice versa.


The first step in the matching procedure is matching single-node cells, i.e., cells of size $1$. For isomorphic graphs, these cells have the same order in the canonical representation of the two graphs. This defines  the initial set of ordered pairs of matched nodes. Then, Algorithm \ref{Alg:Overall} outlines the matching procedure for the other nodes in the graphs, where analogous cells are progressively matched until all nodes are paired (for isomorphic graphs), or the procedure exits when a node could not be matched (for non-isomorphic graphs). 
The algorithm needs multiple iterations because the mapping between nodes is not known a priori, and needs to be progressively estimated.

\begin{algorithm}[ht]
\caption{MP Graph Isomorphism Algorithm}
\label{Alg:Overall}
   \renewcommand{\algorithmicrequire}{\textbf{Input:}}
    \renewcommand{\algorithmicensure}{\textbf{Output:}}
\begin{algorithmic}[1]
\Require $G1(\boldsymbol{\nu}, \boldsymbol{\xi}), G2(\boldsymbol{\nu}, \boldsymbol{\xi})$ (canonical representations of the input graphs)
\Ensure \emph{IsomorphicFlag} (true/false); ${\bf{A}}$ (ordered pairs of the node bijection if the graphs are isomorphic)
\State \If{ single-node cell structures are identical}
\State Construct the initial input pattern ${\bf{p}_0}$
\State Generate the output signature  ${\bf{g}}_1^{(0)}= \Phi({\bf{p}}_0,G1)$   
\State Generate the output signature ${\bf{g}}_2^{(0)}= \Phi({\bf{p}}_0,G2)$
\If {${\bf{g}}_2^{(0)} == {\bf{g}}_1^{(0)}$}
\State Initialize ${\bf{A}}, \overline{{\bf{A}}_1}, \overline{{\bf{A}}_2}$ from all single-node cells
\Else
\State	\emph{IsomorphicFlag} = false
\State        return
\EndIf
\Else
\State	\emph{IsomorphicFlag} = false
\State        return
\EndIf
\State
\While {$\overline{{\bf{A}}_1}$ is not empty}
\State Pick one node $\nu_i \in \overline{{\bf{A}}_1}$ 
\State Construct input pattern ${\bf{p}_1} = \eta(\nu_i , {\bf{A}})$ 
\State Generate the output signature  ${\bf{g}}_1= \Phi({\bf{p}}_1,G1)$
\For {every $\nu_j \in \overline{{\bf{A}}_2}$ such that $\nu_j \sim \nu_i$}
\State Construct input pattern ${\bf{p}_2}$ as a permutation of  ${\bf{p}_1}$
\State Generate the output signature  ${\bf{g}}_2= \Phi({\bf{p}}_2,G2)$
\If {${\bf{g}}_2 == {\bf{g}}_1$}
\State ${\bf{A}} = {\bf{A}} \cup \{(\nu_i,\nu_j)\}$
\State $\overline{{\bf{A}}_1} = \overline{{\bf{A}}_1} -\{\nu_i\}$
\State $\overline{{\bf{A}}_2} = \overline{{\bf{A}}_2} -\{\nu_j\}$  
\State {\bfseries break}  
\EndIf
\EndFor
\If {$\nu_i$ is not matched}
\State \emph{IsomorphicFlag} = false
\State return
\EndIf
\EndWhile
\State \emph{IsomorphicFlag} = true
\end{algorithmic}
\end{algorithm}

The input patterns are  designed to resolve ambiguity in nodes mapping if the graphs are isomorphic. Note that, a distinct input pattern is generated for each node in either graph, and only nodes within analogous cells of the two graphs are investigated for matching. The input pattern is designed such that matched nodes see the same input pattern at analogous entries, and whenever ambiguity exists ambiguous nodes are assigned the same numerical value. Hence, the input patterns should fulfill the following requirements:
\begin{itemize}
\item [\bf{[I1]}]  Input entries for each pair of already matched nodes have the same numerical value, and this value is unique within the input pattern.
\item [\bf{[I2]}]  Input entries that correspond to different cells always have different values, such that  cell ambiguity is  resolved.
\item [\bf{[I3]}] The input entry for the node under investigation is distinct, and the same numerical value is generated for the node in the other graph that is investigated for matching.
\end{itemize}
The generation of the input pattern is described in details in the following section. 

The core component of the graph matching algorithm is the mapping of an input pattern to an output signature, which is the mapping function $\Phi$ in Algorithm \ref{Alg:Overall}. The matching is performed on the output signatures of the two graphs; hence, the mapping should uniquely reflect the graph edge structure. By close inspection of  Algorithm \ref{Alg:Overall}, it is straightforward to deduce that the following conditions are sufficient for the completeness of the algorithm:
\begin{itemize}
\item [\bf{[C1]}] Isomorphic graphs produce identical signatures for each pair of matched nodes.
\item [\bf{[C2]}] Non-Isomorphic graphs do not produce the same signature  for at least one node. 
\item [\bf{[C3]}] A pair of nodes in isomorphic graphs that cannot be matched do not produce identical signatures.
\item [\bf{[C4]}] If multiple bijections exist between two graphs, the algorithm produces at least one of these bijections.
\end{itemize}
The proposed mapping function resembles the message-passing algorithm in section \ref{sec:MP}, and it is described in details in section \ref{sec:mapping}, and a specific  implementation is described in section \ref{sec:implementation}. In section \ref{sec:proofs}, the conditions for satisfying {\bf{C1}}-{\bf{C4}} are derived and proved.


\subsection{Input Pattern Generation}\label{input_gen}

The first input pattern, ${\bf{p}_1} = \eta(\nu_i,{\bf{A}}) \in  \mathbb{R}^M$, is generated as follows (where $\{\alpha_r\}$ and $\beta$ are predefined constants with $\alpha_i \neq \alpha_j \text{ if }  i\neq j,  \beta \neq \alpha_r \ \ \forall r$, and $\beta \neq 0$):
\begin{enumerate}
\item set ${\bf{p}_1}(i) = \beta$.
\item $\forall$ $k \neq i$ such that $\nu_k \in \overline{{\bf{A}}_1}$, set ${\bf{p}_1}(k) = \alpha_r$ where $\nu_k$ is in the $r$-th cell.
\item  $\forall$ $k \neq i$ such that $\nu_k \notin \overline{{\bf{A}}_1}$ (i.e., $\nu_k$ has been matched), set ${\bf{p}_1}(k) = \gamma_k$, where $\gamma_k$ is a random number, and $\gamma_k \neq \gamma_j$ for $k\neq j$,  and $\gamma_k \neq \alpha_r \text{ for all }k, r $.
\end{enumerate}
Similarly, ${\bf{p}_2} = \eta(\nu_j,{\bf{A}})$ is generated  as a permutation of ${\bf{p}_1}$:
\begin{enumerate}
\item set ${\bf{p}_2}(j)  = \beta$.
\item $\forall$ $k \neq j$ such that $\nu_k \in \overline{{\bf{A}}_2}$, set ${\bf{p}_2}(k) = \alpha_r$ where $\nu_k$ is in the $r$-th cell. 
\item The entries that correspond to matched nodes are set to the same random variable of the corresponding entry in ${\bf{p}_1}$. For example, if $(\nu_k, \nu_m) \in {\bf{A}}$, then ${\bf{p}_2}(m) = {\bf{p}_1}(k) = \gamma_k$.
\end{enumerate}
It is straightforward to show that this construction ensures that the requirements {\bf{I1}}-{\bf{I3}} are satisfied.

\subsection{The Mapping Function}\label{sec:mapping}
The design of the mapping functions resembles a message-passing procedure of the belief-propagation algorithm \cite{Koller}. 
Denote the set of edges in the canonical bipartite graph by $\bf{E}$ (not to be confused by the nodes that represent edges of the original graph denoted by $\boldsymbol{\xi}$). 
The function is outlined  in Algorithm \ref{Alg:Mapping}. The sim\_sort() procedure at line \ref{sortline} performs a separate sort of the entries of each cell. To avoid permutation errors, $f^{(l)}$ and $y^{(l)}$ are either applied to the \emph{ordered} vector of $\{\lambda(E)\}$ or to the  sum of its elements as in conventional belief-propagation \cite{Richardson}. 

\begin{algorithm}[ht]
\caption{MP-Based Mapping Function}
\label{Alg:Mapping}
   \renewcommand{\algorithmicrequire}{\textbf{Input:}}
    \renewcommand{\algorithmicensure}{\textbf{Output:}}
\begin{algorithmic}[1]
\Require $G(\boldsymbol{\nu}, \boldsymbol{\xi})$; Input Pattern ${\bf{p}} \in \mathbb{R}^M$
\Ensure signature ${\bf{g}} \in \mathbb{R}^{ML}$
\State Set $\Gamma(\boldsymbol{\nu}) = \bf{p}$; $\lambda({\bf{E}}) = \bf{0}$  
\For { $l  = 1:L $}
\For {every $\nu \in {\boldsymbol{\nu}}$}
\For{every $E \in {\bf{Q}}(\nu)$}
\State $\Psi(E) = f^{(l)}\left(\{\lambda(E^\prime)\}_{E^\prime \in {\bf{Q}}(\nu)-\{E\}},\Gamma(\nu) \right)$ 
\EndFor
\EndFor
\For {every $\xi \in {\boldsymbol{\xi}}$}
\For{every $E \in {\bf{Q}}(\xi)$}
\State $\lambda(E) = h^{(l)}\left(\Psi(E), \Psi(\tilde{E}) \right)$ \label{line:lambda}  \label{lambda_def} 
\EndFor
\EndFor
\For {every $\nu \in {\boldsymbol{\nu}}$}
\State $\Gamma(\nu) = y^{(l)}\left(\{\lambda(E^\prime)\}_{E^\prime \in {\bf{Q}}(\nu)}\right)$ 
\EndFor
\State ${\bf{g}}((l-1)M+1:lM)= \text{sim\_sort}\left(\Gamma(\boldsymbol{\nu})\right)$ \label{sortline}
\EndFor
\end{algorithmic}
\end{algorithm}

\section{Algorithm Analysis}\label{sec:analysis}

\subsection{Conditions}\label{sec:proofs}
The first step in the algorithm is matching the cell structure of the two graphs. If they do not match, then the graph are not isomorphic. Therefore, in the following discussion it is always assumed that the two graphs have identical cell structure. In the following, we prove that the sufficient conditions {\bf{C1}}-{\bf{C4}} (as outlined in section \ref{sec:outline}) are satisfied.

By close inspection of the matching algorithm, it is straightforward to deduce that {\bf{C1}} is always satisfied with any deterministic choice  of the mapping function. This is self-evident since the same numerical entries of the input propagate through the same deterministic mapping. 
The other conditions are less obvious and require extra conditions on the input pattern and the mapping functions. We start with a general observation of the matching algorithm that will be used in subsequent proofs.

\begin{lemma} \label {lemma:matching}
If a statement about the matching algorithm is true when the maximum cell size is less than $n$, then it is also true when the maximum cell size is $n$. 
\end{lemma}
\begin{proof}
A cell of size $n$ could be partitioned into two cells of sizes $1$ and $n-1$ (with a total of $n$ possible permutations). Therefore, matching a cell $\theta^{(1)} \in \text{G1}$ of size $n$  to the corresponding cell $\theta^{(2)}\in \text{G2}$ is equivalent to running $n$ matchings of two cells of size $1$ and $n-1$ between a partition of $\theta^{(1)}$ and all possible partitions of $\theta^{(2)}$. The two cells, $\theta^{(1)}$ and $\theta^{(2)}$,  are matched if there is at least one exact matching among the $n$ matchings, otherwise they are not matched. Hence, the problem is factored to matching cells of size less than $n$ where the statement is true. Therefore, the statement would also be true when the cell size is $n$.
\end{proof}
The above lemma simplifies proving statements about the matching algorithm to the  case where each cell has the minimum number of nodes because the lemma completes the mathematical induction procedure if the induction is done over the maximum cell size.
\begin{theorem}\label{theorem:I}
A sufficient condition for {\bf{C2}} is that  the constituent functions $\{f^{(l)}, h^{(l)}, y^{(l)}\}_{1\leq l \leq L }$ are injective.
\end{theorem}
\begin{proof}
Let $\text{G1}$ and $\text{G2}$ be non-isomorphic graphs with identical cell structure. The proof proceeds by induction on the size of the biggest cell in the graph. 
First, assume that all  cells have only one node, so that all entries of the input pattern are distinct (by following the construction procedure in section \ref{input_gen}). Hence, if $f^{(1)}$ is injective, then all $\{\Psi(E)\}$ and $\{\Psi(\tilde{E})\}$ in line \ref{lambda_def} of Algorithm \ref{Alg:Mapping}  have distinct values for each $E$ and $\tilde{E}$. 
Similarly, if $\{f^{(l)}, h^{(l)}, y^{(l)}\}_{1\leq l \leq L }$ are injective, then $\{\Gamma(\nu)\}$ at each iteration will always have distinct values that depends on the inputs of the mapping functions. Hence, the output signature will be composed of distinct values that reflect the edge structure. If the two graphs are not isomorphic, then at least one node in $\text{G1}$ does not have a mapping to a node in $\text{G2}$ that preserves the edge structure. Therefore, at least one mapping function has different inputs for the two graphs which would produce different outputs that propagate as different entries in the output signature. This different mapping can always be reached if a sufficient number of iterations is utilized so that the longest graph route is fully covered. The proof is completed for the larger cell size by invoking Lemma \ref{lemma:matching}.
\end{proof}

\begin{corollary} \label{corollary:C3}
If the input pattern is constructed as in section \ref{input_gen}, then a sufficient condition for {\bf{C3}} is that  the constituent functions $\{f^{(l)}, h^{(l)}, y^{(l)}\}_{1\leq l \leq L }$ are injective.
\end{corollary}
\begin{proof}
Note that, nodes from different cells are not matched in Algorithm \ref{Alg:Overall}. Therefore, {\bf{C3}} is relevant only when the  cell size is at least $2$. Therefore, by utilizing Lemma \ref{lemma:matching}, we need to prove the corollary only when the maximum cell size is $2$. Let $\theta^{(1)}$ and $\theta^{(2)}$ be analogous cells of the isomorphic graphs $\text{G1}$ and $\text{G2}$ respectively. Let the nodes in each cell be $\{\nu^{(i)}, \eta^{(i)}\}$ with $i\in\{1,2\}$. Further, assume $\nu^{(1)}$ is matched to $\nu^{(2)}$ but not matched to $\eta^{(2)}$. Hence, in inspecting the matching between $\nu^{(1)}$ and $\eta^{(2)}$, different entries of the input pattern are produced at the matched nodes $\nu^{(1)}$ and $\nu^{(2)}$. If all the constituent functions are injective, this would result in a different output at some stage that would propagate to the output and produce mismatched signatures. 
\end{proof}
Finally, {\bf{C4}} is fulfilled by the following result:
\begin{lemma} \label {lemma:C4}
The sufficient conditions for {\bf{C4}} are:
\begin{enumerate}
\item Entries in the input patterns to the two graphs that correspond to matched nodes are identical.
\item Entries in an input pattern that correspond to already matched nodes have distinct values.
\item $\{f^{(l)}, h^{(l)}, y^{(l)}\}_{1\leq l \leq L }$ are injective.
 \end{enumerate}
\end{lemma}
\begin{proof}
As in Corollary \ref{corollary:C3}, we need only to study the case when analogous cells of the two graphs have size $2$. Let $\theta^{(1)}$ and $\theta^{(2)}$ be analogous cells of the isomorphic graphs $\text{G1}$ and $\text{G2}$ respectively. Let the nodes in each cell be $\{\nu^{(i)}, \eta^{(i)}\}$ with $i\in\{1,2\}$. Further, assume $\nu^{(1)}$ can be matched to either $\nu^{(2)}$ or $\eta^{(2)}$, and similarly for $\eta^{(1)}$. Hence, we have exactly two bijections $\{(\nu^{(1)}, \nu^{(2)}) , (\eta^{(1)}, \eta^{(2)})\}$ or $\{(\nu^{(1)}, \eta^{(2)}) , (\eta^{(1)}, \nu^{(2)})\}$. If $\nu^{(1)}$ and $\nu^{(2)}$ are inspected first and the input pattern construction satisfies the first condition, they would produce identical signatures; and hence declared as matched pairs. After matching $\nu^{(1)}$ and $\nu^{(2)}$, only $\eta^{(2)}$ could be inspected for matching $\eta^{(1)}$. By the second condition, the input patterns have identical values at the entries that correspond to  $\eta^{(1)}$ and $\eta^{(2)}$ and identical values at the entries that correspond to $\nu^{(1)}$ and $\nu^{(2)}$. Therefore, the two graphs would produce identical signatures, and  $\eta^{(1)}$ and $\eta^{(2)}$ would be matched. 
\end{proof}

\subsection{Implementation}\label{sec:implementation}

The  difficulty with designing injective functions for the matching algorithm is that the function inputs can have arbitrary order  if the maximum cell size is bigger than $1$, because of the arbitrary ordering of nodes within a cell. This difficulty can be mitigated if the function is applied to either a sorted vector of the inputs or the sum of all inputs.  In our implementation,  the constituent functions are chosen to resemble  the sum-product algorithm \cite{Richardson}. Although the functions are not strictly injective, it is with probability $0$ to have a similar output for different inputs.  The constituent functions at iteration $l$ are defined as:
\begin{changemargin}{-0.5cm}{0cm}
\begin{eqnarray}
f^{(l)} \left(E, \nu\right) &\triangleq& \Gamma(\nu) + a^{(l)} \sum_{E^\prime \in {\bf{Q}}(\nu)-\{E\}}  \lambda(E^\prime) \label{eq:fl}\\ 
h^{(l)} \left(E, \tilde{E} \right) &\triangleq&  b^{(l)} \Psi(E) + c^{(l)} \Psi(\tilde{E})   \label{eq:E_message} \\
y^{(l)}(\nu) &\triangleq& d^{(l)} \sum_{E \in {\bf{Q}}(\nu)}\lambda(E) \ \ \ \ \ \label{eq:yl}
\end{eqnarray}
\end{changemargin}
where the scalars $\{a^{(l)}, b^{(l)}, c^{(l)}, d^{(l)}\}$ are different and have different nonzero values at different iterations.
To avoid numerical precision issues in computing the output signature, integer data types are used rather than floating-point numbers. 

\subsection{Supervised Matching} \label{sec:supervised} 
Algorithm \ref{Alg:Overall} is an exhaustive search procedure that iteratively investigates all possible isomorphism mappings. Therefore, it represents the worst-case complexity of the matching algorithm. 
The number of iterations in Algorithm \ref{Alg:Overall} could be significantly reduced by noting that there is no mapping ambiguity  if the values of the  output fingerprint  are \emph{distinct} at any point of time. This also applies at a cell-level matching, i.e., if the fingerprint portion that corresponds to a particular cell has distinct values, then all nodes of this cell could be mapped unambiguously. 

By exploiting the above observation, the number of matching iterations are significantly reduced by \emph{supervising} the matching procedure such that each \emph{sub}-fingerprint that corresponds to an unmatched cell is investigated at each  iteration of Algorithm \ref{Alg:Mapping}. If it does not have duplicate values, then this sub-fingerprint is compared with the corresponding sub-fingerprint of the other graph. If they are identical, then the two cells are matched with the direct mapping between nodes  that correspond to the distinct values. Further, If these analogous sub-fingerprints are not identical, then these two cells cannot be matched and the two graphs are declared  non-isomorphic. This  supervised matching could be further refined by comparing non-duplicates within each sub-fingerprint without requiring the the whole sub-fingerprint to be distinct.

\subsection{Complexity} \label{sec:complexity}
The algorithm has two components: the construction of the canonical forms, and the matching of cells. The construction of the canonical forms requires computing the order of each node and sorting the nodes according to their order and their neighbors order. This in general has a complexity of  $O(M \log M)$. This complexity component persists regardless of the  heuristics in the matching procedure. Hence, it represents the best-case complexity of the overall algorithm.

The worst-case complexity of the matching algorithms follows when the exhaustive procedure in Algorithm \ref{Alg:Overall} is utilized without an improvement from the supervised matching.
The worst-case number of iterations in Algorithm \ref{Alg:Overall} is $M^2$, which takes effect when all nodes belong to a single cell.  The worst number of iterations in Algorithm \ref{Alg:Mapping} is $L=M-1$. With the  choice of the constituent functions as in section \ref{sec:implementation}, the complexity in each iteration in Algorithm \ref{Alg:Mapping} is $O(K + D \log D)$, where $D$ is the maximum cell size. Note that, $O( D \log D)$ is the worst case complexity of $\text{sim\_sort}()$ procedure at the end of each iteration.  Therefore,  graphs with small cells would in general gave $D \ll K$  and the overall complexity becomes $O(M^3K)$. For graphs with large cells $D$ approaches $M$, and the overall complexity is approximately  $O(M^4 \log M)$.

The above complexity is the worst-case complexity for isomorphic graphs. For non-isomorphic graphs, the search time is much less because a small number of iterations is expected before a mismatch of fingerprints occurs. Similar behavior is expected with supervised matching of isomorphic graphs. In both cases, the complexity is dominated by the canonical forms construction, i.e.,  $O(M \log M)$.

  The polynomial complexity of the matching algorithm stems primarily from the absence of backtracking in the search procedure. Backtracking is not needed because the generation of the input pattern mitigates ambiguity in the isomorphic mapping. Hence, a mismatch of fingerprints at any stage implies that the graphs are non-isomorphic.

\section{Discussion}\label{sec:discussion}

\subsection{Generalizations} \label{sec:generalization}
\subsubsection{Weighted and Directed Graphs} To generalize the proposed algorithm to weighted graphs, the weights on weighted graph should be incorporated in both constructing the canonical form, and computing  edge messages in \eqref{eq:E_message}. In this case, the classification of cells  not only counts the number of edges but also utilizes the weights on each edge such that edges in the same cell have edges with exactly the same weight.
Directed graphs could be handled by redefining the canonical graph such that one partition corresponds to the source nodes and the other partition corresponds to the end nodes, i.e., the end nodes resemble edges in the  canonical representation for undirected graphs. In this case,  nodes with ongoing and outgoing edges are represented twice at both sides of the bipartite canonical representation. Therefore, a unique numerical identifier is associated with the same node at both sides to remove node ambiguity. This unique identifier is utilized in the message generation \eqref{eq:fl}-\eqref{eq:yl}.

\subsubsection{Graph Automorphism}
The proposed algorithm can be straightforwardly extended to the graph automorphism problem, by restricting the mapping search within a cell to non-identical nodes, i.e., for a cell of $k$ nodes, the first node to be matched would be compared to the other $k-1$ nodes and so on. This is done for at least one cell in the matching procedure. Other automorphisms could be extracted similarly be restricting the matching to nodes that were not paired in earlier automorphisms. Note that, the same idea of constrained mapping could be used to find other bijections in the original graph isomorphism problem.

\subsection{Experimental Evaluation}
The proposed algorithm has been evaluated using  test graphs from the TC-15 database \cite{foggia2001database} which is available  at \cite{databasedownload}. A total of 12,200 pairs of isomorphic graphs, as outlined in Table \ref{table:test_graphs},  were successfully evaluated. The maximum graph size is $1296$ nodes and the minimum graph size is $16$ nodes. 

\begin{table}[h]
\centering
\caption{Test graphs in TC-15 graph database}
\begin{tabular}{|c|c|c|}
  \hline
   Prefix & Description & No. of test\\
   & &  graph pairs \\
   \hline
   \text{iso\_r$\rho$}& randomly generated graphs  & $$\\
   &with $\rho = 01, 001, 005$ &$3\times 10^3$\\
   \hline
   \text{iso\_m2D$\rho$}& bi-dimensional meshes  & \\
   &with $\rho = \ , r2, r4, r6$& $4\times 10^3$\\
   \hline
   \text{iso\_m3D$\rho$}& tri-dimensional meshes  & \\
   &with $\rho = \ , r2, r4, r6$& $3.2\times 10^3$\\
   \hline
   \text{iso\_m4D$\rho$}& quadri-dimensional meshes  & \\
   &with $\rho = \ , r2, r4, r6$& $2\times 10^3$\\
   \hline
    \end{tabular}
\label{table:test_graphs}
\end{table}


In Fig. \ref{fig:processing_time}, the median processing time for all isomorphic graphs in the TC-15 database are shown\footnote{Executed on a MacPro with 2.7GHz dual-core Intel i5 microprocessor and 8GB of memory.}, where each data point corresonds to $100$ pairs of isomorphic graphs. All the graphs are matched successfully, and in most case the complexity is upper bounded by $O(M^2)$.
\begin{figure}[ht]
\begin{changemargin}{-0.2cm}{-1cm}
\begin{center}
\includegraphics[width=3.8in,height=3.3in, trim = 0.35in 0.53in 0.0in 0.0in, clip,]{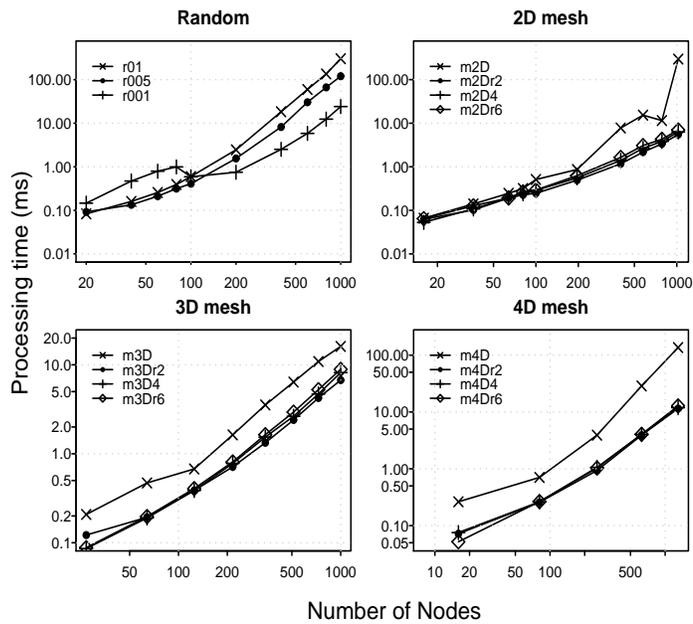}
\end{center} 
\vspace{-2mm}
\caption{Median processing time of the matching algorithm for graphs in the TC-15  database}.
\label{fig:processing_time}
\end{changemargin}
\end{figure}

%
In Fig. \ref{fig:comparison} the median processing time of the proposed algorithm is compared with the Nauty matching algorithm \cite{mckay1981practical}\footnote{The  Nauty implementation in the software package at \cite{nautydownload} is used.}, which represents an optimized algorithm from both design and software implementation perspectives. The Nauty algorithm has noticeably lower complexity than the message-passing algorithm for the test graphs. However, as mentioned earlier it could have an exponential complexity in the worst case.
Further, the optimization of data structures and software implementation of the message-passing algorithm has not been addressed.


\begin{figure}[ht]
\begin{center}
\includegraphics[width=2.7in,height=2.2in, trim = 0.35in 0.53in 0.0in 0.0in, clip,]{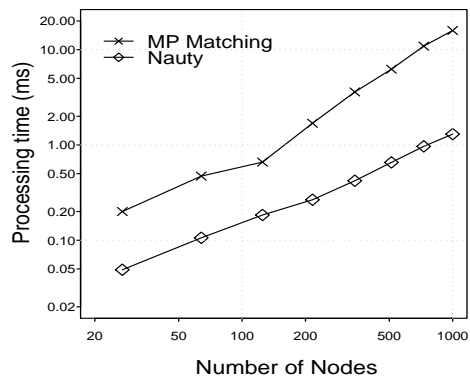}
\end{center} 
\vspace{-4mm}
\caption{Comparison of the median processing times with Nauty algorithm for Tri-Dimensional meshes (m3D)}.
\label{fig:comparison}
\end{figure}

\section{Conclusion}
The work represents a new approach for graph isomorphism and automorphism without the need of backtracking, which could result in exponential complexity in the worst case.
The main idea of  message-passing is borrowed from different disciplines that  employed it to  curb the complexity of the maximum-likelihood estimator given a general observation model. The proposed algorithm formulates the graph isomorphism problem in the same framework, by introducing a bipartite graph representation and a special design of the input pattern, which removes the ambiguity in the isomorphism mapping. We introduced sufficient conditions for completeness and uniqueness, and introduced few heuristics that significantly reduced the number of iterations in the search algorithm. 
The algorithm does not assume a particular structure or set  constraints on the test graphs. 
This favorable performance is primarily because of the iterative search procedure that does not require backtracking. This performance has been achieved by earlier algorithms that are based on random-walks if the input graphs satisfy some spectral constraints. 

Future work also includes generalizing the procedure to the subgraph isomorphism problem and inexact graph matching.


%

\bibliographystyle{siam}
\bibliography{graph_ref}
\end{document}